\documentclass{ifacconf}

\usepackage{natbib}        

\usepackage{epsfig} 
\usepackage{subfig}
\usepackage{graphicx}    
\usepackage{epstopdf} 
\usepackage{amsmath} 
\usepackage{amssymb}  
\usepackage{multirow}
\usepackage[dvipsnames]{xcolor}
\usepackage{soul}
\usepackage[cmintegrals]{newtxmath}
\usepackage{float}
\usepackage{mathrsfs}
\usepackage{algorithmic}
\usepackage{algorithm}
\usepackage{subdepth}

\renewcommand{\ge}{\geq}

\renewcommand{\le}{\leq}

\newcommand\T{\rm T}
\newcommand{\mR}
{\mathbb{R}}

\usepackage{multirow}
\usepackage{color}

\DeclareMathAlphabet\mathbfcal{OMS}{cmsy}{b}{n}

\definecolor{cover}{RGB}{72,0,255}
\definecolor{gold}{RGB}{204,163,0}
\definecolor{darkblue}{RGB}{0,21,120}
\definecolor{rred}{RGB}{225,81,18}
\definecolor{dred}{RGB}{170,60,10}
\definecolor{bblue}{RGB}{0,115,189}
\definecolor{ggreen}{RGB}{0,160,100}
\definecolor{dgreen}{RGB}{0,99,61}
\definecolor{lblue}{RGB}{77,191,237}

\definecolor{d2red}{RGB}{215,48,39}
\definecolor{d2orange}{RGB}{252,141,89}
\definecolor{d2yellow}{RGB}{254,224,144}
\definecolor{d2lblue}{RGB}{224,243,248}
\definecolor{d2blue}{RGB}{145,191,219}
\definecolor{d2dblue}{RGB}{69,117,180}

\newtheorem{Lemma}{Lemma}
\newtheorem{Theorem}{Theorem}
\newtheorem{corollary}{Corollary}
\newtheorem{definition}{Definition}

\newtheorem{Remark}{Remark}

\begin{document}

\begin{frontmatter}

\title{Performance Analysis of Time-Delay Systems
under External Perturbations Using Output-to-Output Gain}

\author[UU]{Ruslan Seifullaev},
\author[UU]{Andr\'{e} M. H. Teixeira}

\address[UU]{Division of Systems and Control, Department of Information Technology, 
Uppsala University, Sweden (e-mail: \{ruslan.seifullaev, andre.teixeira\}@it.uu.se).}

\begin{abstract}
Communication delays are inherent in networked control systems and may significantly affect closed-loop performance. This paper investigates the impact of external perturbations in observer-based linear systems with delayed measurements and control signals using the output-to-output gain (OOG) framework. By employing dissipativity theory and Lyapunov--Krasovskii functionals, delay-dependent linear matrix inequality conditions are derived that provide explicit upper bounds on the OOG for systems with independent delays in measurement and actuation channels. 
In addition, two special classes of constant-delay systems are considered: Pad\'e-approximated systems and finite-dimension reducible systems. Numerical examples illustrate the applicability of the proposed approaches.
\end{abstract}


\end{frontmatter}


\section{Introduction}

\begin{figure*}[h!]
\centering
\includegraphics[scale=0.4]{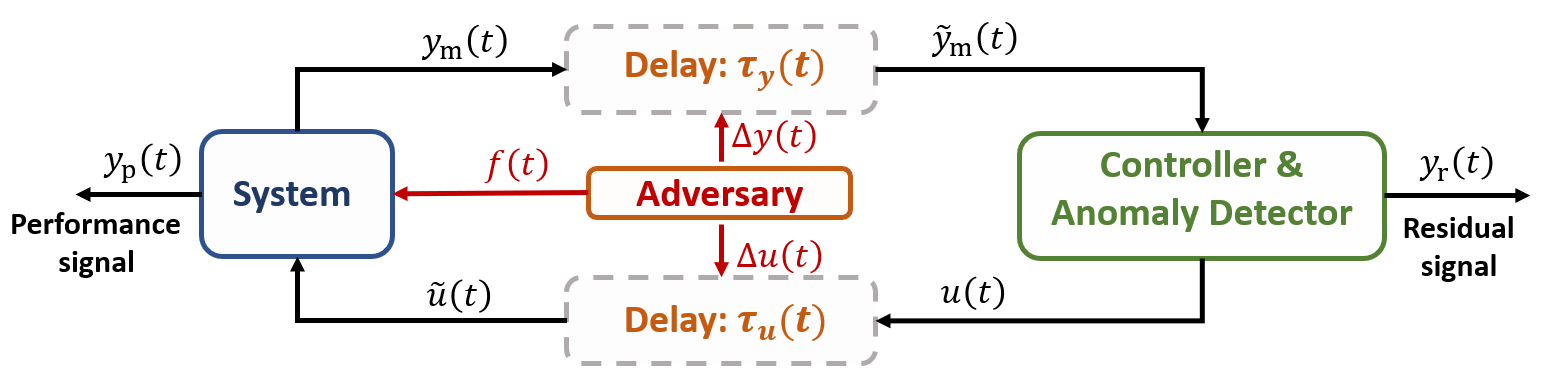}
\caption{A time-delay system under external perturbations}
\label{diagram}
\end{figure*}

Networked control architectures are increasingly employed in modern engineering systems, including power networks, industrial automation, transportation systems, and cyber-physical infrastructures. In many of these applications, sensors, actuators, and control units are distributed and interconnected through shared communication networks. Such networked control systems (NCSs) provide substantial benefits in terms of flexibility, scalability, and reduced implementation costs. At the same time, the presence of communication constraints may significantly affect closed-loop stability and performance. Among all network phenomena, communication delays constitute one of the most important and widely studied challenges in NCSs. Delays naturally arise due to sensing, routing, scheduling, and transmission processes and may affect both measurement and actuation channels. Even when the underlying plant is linear and nominally stable, delays may degrade transient performance, reduce robustness margins, and potentially destabilize the closed-loop system. Consequently, the analysis and control of time-delay systems have attracted considerable attention over the past decades; see, e.g., \cite{Fridman_book,Kharitonov13,chen,Richard03} and the references therein.

In addition to network-induced imperfections, modern NCSs are exposed to a variety of faults, disturbances, and malicious activities. Hardware degradation, sensor malfunctions, software errors, communication failures, or externally generated perturbations may alter the information exchanged between different components of the control system. Such perturbations may affect sensor measurements, actuator signals, or directly influence the physical process itself. While many abnormal operating conditions can be detected using conventional monitoring schemes, certain perturbations may remain weakly visible to residual-based detectors while still causing significant degradation of closed-loop performance. Understanding the worst-case impact of such perturbations is therefore an important problem from both reliability and security perspectives.

A large body of literature has addressed fault diagnosis, fault-tolerant control, and security of networked control systems; see, e.g., \cite{Cardenas08,Sandberg15,ZHAO2020109128,24_Mousavinejad} and references therein. Typical performance and robustness measures include the $H_\infty$ norm, \cite{Zhou96}, induced-gain analysis, and fault sensitivity indices such as the $H_-$ index, \cite{WANG07}. Although these metrics provide valuable information about disturbance amplification or fault detectability, they generally evaluate impact and detectability separately. As a result, they do not directly quantify the trade-off between maintaining a small residual signal and simultaneously causing a large degradation of system performance.

To address this limitation, the output-to-output gain (OOG) was introduced in \cite{Teixeira2015}. The OOG quantifies the worst-case trade-off between performance degradation and detectability by characterizing the maximal energy of the performance output that can be generated while keeping the residual output energy bounded. In this way, the metric captures the fundamental objective of stealthy perturbations: maximizing performance degradation while remaining difficult to detect. Existing OOG results, however, have been developed primarily for finite-dimensional delay-free systems. To the best of the authors' knowledge, corresponding results for time-delay systems have not yet been established.

The extension of OOG analysis to delayed systems is nontrivial. Unlike ordinary differential equations, time-delay systems are inherently infinite-dimensional, since their future evolution depends on the entire history of the state trajectory over the delay interval. Consequently, techniques developed for finite-dimensional systems cannot be applied directly. Delay-dependent analysis typically relies on Lyapunov--Krasovskii functionals and functional differential equation theory, leading to significantly different mathematical tools and performance conditions.

Motivated by these observations, this paper investigates performance degradation in observer-based networked control systems with communication delays and external perturbations. The considered framework includes independent delays in measurement and actuation channels, additive perturbations entering communication links, and physical disturbances acting directly on the plant dynamics. The objective is to quantify the worst-case amplification from the residual output to the performance output through the output-to-output gain.

To this end, we extend the dissipativity-based OOG framework to time-delay systems. Using Lyapunov--Krasovskii functionals, the descriptor approach, and reciprocally convex inequalities, we derive delay-dependent linear matrix inequality (LMI) conditions guaranteeing output strict dissipativity and yielding explicit upper bounds on the OOG. In addition, we discuss two special classes of constant-delay systems for which finite-dimensional representations can be employed: systems approximated using Pad'e delay models and finite-dimension reducible systems in the sense of \cite{Churilov13}.

The main contributions of this paper can be summarized as follows:
\begin{itemize}
\item {\it Output-to-output gain analysis for time-delay systems}. We extend the OOG framework to observer-based linear systems with communication delays and formulate the corresponding performance analysis problem in an infinite-dimensional setting.

\item {\it Delay-dependent dissipativity conditions}. By employing Lyapunov--Krasovskii functionals together with descriptor and reciprocally convex techniques, we derive computationally tractable LMI conditions that provide explicit upper bounds on the OOG.

\item {\it Finite-dimensional reductions}. We investigate two special cases involving constant delays. For Pad'e-approximated systems, existing delay-free OOG results can be applied directly. For finite-dimension reducible systems, we establish conditions under which the delayed dynamics admit an exact finite-dimensional representation and derive corresponding OOG estimates.
\end{itemize}

The remainder of the paper is organized as follows. Section~II presents the problem formulation and derives the delayed observer-based closed-loop model. Section~III develops the dissipativity-based OOG analysis and provides delay-dependent LMI conditions. Section~IV discusses special cases based on Pad'e approximations and finite-dimensional reducibility. Section~V presents numerical examples illustrating the proposed methods. Finally, Section~VI concludes the paper.

\section{Problem Formulation}
Consider the following system:
\begin{equation}\label{sys_lin}
\begin{aligned}
\dot x_{\rm p}(t) &= A_{\rm p}x_{\rm p}(t)+B_{\rm p}\tilde u(t),\\
y_{\rm m}(t) &=C_{\rm m,o}x_{\rm p}(t),\quad y_{\rm p}(t) =C_{\rm p,o}x_{\rm p}(t)+D_{\rm p,o}\tilde u(t),
\end{aligned}
\end{equation}
where $x_{\rm p}(t) \in \mR^{n_x}$ is the state vector, $\tilde u(t) \in \mR^{n_u}$ is the control input, $y_{\rm m}(t) \in \mR^{n_y}$ is the measurement output, $y_{\rm p}(t) \in \mR^{n_{y_{\rm p}}}$ is the performance output, and $A_{\rm p}$, $B_{\rm p}$, $C_{\rm m,o}$, $C_{\rm p,o}$ and $D_{\rm p,o}$ are the matrices of appropriate dimensions, the pair $(A_{\rm p}, B_{\rm p})$ is controllable and the pair $(A_{\rm p}, C_{\rm m,o})$ is observable. 

We consider the scenario in which the network channels are subject to additive signal perturbations $\Delta y(t)=\Gamma_{\rm y}\tilde a_{\rm y}(t)$ and $\Delta u(t)=\Gamma_{\rm u}\tilde a_{\rm u}(t)$ acting on the sensor and actuator signals, respectively, possibly inserted by an adversary or originating from an external source. Here the matrices $\Gamma_{\rm y}$ and $\Gamma_{\rm u}$ specify the channels into which the additive signals $\tilde a_{\rm y}(t)\in\mR^{n_{a_{\rm y}}}$ and $\tilde a_{\rm u}(t)\in\mR^{n_{a_{\rm u}}}$ may enter and their structure.
Additionally, we assume that $f(t)=\Gamma_{\rm p} a_{\rm p}(t)$ represents a physical fault signal acting directly on the plant dynamics, where $a_{\rm p}(t)\in\mR^{n_{a_{\rm p}}}$, and $\Gamma_{\rm p}$ is the corresponding fault distribution matrix determining how the disturbance enters the state dynamics, see Figure~\ref{diagram}.

We assume that the communication channels are subject to time-varying delays. In particular, the sensor signal $y_{\rm m}(t)$ is delayed by $\tau_{\rm y}(t)$, while the controller output $u(t) \in \mR^{n_u}$ is delayed by $\tau_{\rm u}(t)$. We also assume that the delays are bounded and slowly-varying, i.e., 
\begin{equation}
\begin{aligned}
\tau_{\rm y}(t)&\in[0,h_{\rm y}], \,\,\, \dot\tau_{\rm y}(t)\le d_{\rm y}, \\
\tau_{\rm u}(t)&\in[0,h_{\rm u}], \,\,\, \dot\tau_{\rm u}(t)\le d_{\rm u},
\end{aligned}
\end{equation}
where $h_{\rm y}>0$, $h_{\rm u}>0$, $d_{\rm y}\le1$, and $d_{\rm u}\le1$.
The external perturbations $\tilde a_{\rm y}(t)$ and $\tilde a_{\rm u}(t)$ may also experience delays along the communication paths. We denote the corresponding received signals by 
$$a_{\rm y}(t)=\tilde a_{\rm y}(t-\tilde\tau_{\rm y}(t)) \quad \mbox{and} \quad a_{\rm u}(t)=\tilde a_{\rm u}(t-\tilde\tau_{\rm u}(t)),$$ respectively\footnote{The delays $\tilde\tau_{\rm y}(t)$ and $\tilde\tau_{\rm u}(t)$ need not coincide with  $\tau_{\rm y}(t)$ and $\tau_{\rm u}(t)$, since the measurement/control signals and the external perturbations may propagate through different communication routes. In this paper, however, we do not address the design of the perturbation signals themselves and focus only on their worst-case impact. Consequently, it suffices to model the received signals $a_{\rm y}(t)$ and $a_{\rm u}(t)$ directly; the exact values of the delays are immaterial, as delayed signals belong to the same signal class as the original inputs.}\!.
Therefore, the signal received by the controller is
\begin{equation}\label{att_additive_y}
\tilde y_{\rm m}(t) = y_{\rm m}(t-\tau_{\rm y}(t)) + \Gamma_{\rm y} a_{\rm y}(t),
\end{equation}
while the signal applied at the actuator is
\begin{equation}\label{att_additive_u}
\tilde u(t) = u(t-\tau_{\rm u}(t)) + \Gamma_{\rm u} a_{\rm u}(t),
\end{equation}
where $u(t)$ is the controller output, see Figure~\ref{diagram}. 

We consider the following observer-based controller structure:
\begin{equation}\label{obs_lin}
\begin{aligned}
\dot{\hat x}_{\rm p}(t) &= A_{\rm p}\hat x_{\rm p}(t)+B_{\rm p}u(t)+Ky_{\rm r}(t),\,\,\,\,
u(t) = L\hat x_{\rm p}(t),\\
\hat y_{\rm m}(t) &=C_{\rm m,o}\hat x_{\rm p}(t),\,\,\,\,
y_{\rm r}(t) = \tilde y_{\rm m}(t)-\hat y_{\rm m}(t),
\end{aligned}
\end{equation}
where $K\in\mR^{n_x\times n_y}$ and $L\in\mR^{n_u\times n_x}$ are the observer and controller gains, respectively.
Define the augmented vectors $
x(t)=  \begin{bmatrix}
x_{\rm p}(t)\\
x_{\rm p}(t)-\hat x_{\rm p}(t)\\
\end{bmatrix}
$ and
$
a(t)=  \begin{bmatrix}
a_{\rm u}(t)\\
a_{\rm y}(t)\\
a_{\rm p}(t)\\
\end{bmatrix}.
$
Then the resulting closed-loop system can be written as follows:
\begin{equation}\label{sys_nonl_cl}
\begin{aligned}
\dot x(t) &= Ax(t)+A_1x(t-\tau_{\rm u}(t))+A_2x(t-\tau_{\rm y}(t))+Ba(t),\\ 
y_{\rm p}(t) &= C_{\rm p}x(t)+C_1x(t-\tau_{\rm u}(t))+D_{\rm p}a(t),   \\
y_{\rm r}(t) &=C_{\rm r}x(t)+C_2x(t-\tau_{\rm y}(t))+D_{\rm r}a(t),
\end{aligned}
\end{equation}
where
\begin{equation*}
\begin{aligned}
 A & = \begin{bmatrix}
A_{\rm p}&0\\
-B_{\rm p}L+KC_{\rm m,o}\,\,\,&A_{\rm p}+B_{\rm p}L-KC_{\rm m,o}\\
\end{bmatrix},\\
A_1  &= \begin{bmatrix}
B_{\rm p}L\,\,\,&-B_{\rm p}L\\
B_{\rm p}L\,\,\,&-B_{\rm p}L\\
\end{bmatrix},\,\,
A_2  = \begin{bmatrix}
0&0\\
-KC_{\rm m,o}&0\\
\end{bmatrix},\\
B &= \begin{bmatrix}
B_{\rm p}\Gamma_{\rm u}&0&\Gamma_{\rm p}\\
B_{\rm p}\Gamma_{\rm u}&-K\Gamma_{\rm y}&\Gamma_{\rm p}\\
\end{bmatrix}, \,\,D_{\rm p}  = \left[D_{\rm p,o}\Gamma_{\rm u},\, 0,\, 0\right],\\
C_{\rm p}  &= \left[C_{\rm p,o},\, 0\right], \,\,C_1 = \left[D_{\rm p,o}L,\, -D_{\rm p,o}L
\right],\\
C_{\rm r}  &= \left[-C_{\rm m,o},\, C_{\rm m,o}\right], \,\,C_2 = \left[C_{\rm m,o},\, 0
\right],\,\,D_{\rm r}  = \left[0,\, \Gamma_{\rm y},\, 0\right].
\end{aligned}
\end{equation*}

Thus, we obtain a time-delay system with two outputs, $y_{\rm p}$ and $y_{\rm r}$, driven by the external input $a$.
In practice, the residual signal $y_{\rm r}$ is commonly used in monitoring and anomaly detection schemes, which raise an alarm whenever the residual energy $||y_{\rm r}||^2_{L_2}$ exceeds a prescribed threshold. Consequently, perturbations that keep the residual energy small may remain weakly visible to such monitoring mechanisms while still significantly degrading closed-loop performance.
Motivated by this observation, the main objective of this paper is to characterize the worst-case impact of external perturbations in terms of the maximal achievable energy of the performance output $y_{\rm p}$ under constraints on the residual output energy $||y_{\rm r}||_{L_2}^2$.

\section{Main result: performance analysis with the output-to-output gain}
Since the closed-loop system \eqref{sys_nonl_cl} contains time-varying delays, it is inherently infinite-dimensional. 
Consequently, its state at time $t$ is characterized not only by the current value $x(t)$ but by the entire state trajectory over the delay interval. 
Following the standard framework for functional differential equations, see \cite{Fridman_book,chen}, we define the state segment
\[
x_t(\theta)=x(t+\theta), \qquad \theta\in[-h,0],
\]
associated with the trajectory $x(t)$, 
where
$
h=\max\{h_{\rm u},h_{\rm y}\}
$
is the maximal admissible delay. 
For the subsequent Lyapunov--Krasovskii analysis, we assume that the state segments $x_t$ belong to the function space
$W$, which denotes the Banach space of  absolutely continuous functions on $[-h,0]$ whose derivatives are square integrable. The space $W$ is equipped with the norm
$
\|z(\theta)\|_W=
\max_{\theta\in[-h,0]}
|z(\theta)|
+
\left(
\int_{-h}^{0}
|\dot z(s)|^2 ds
\right)^{\frac12}.
$

Without loss of generality, we extend the function $x(t)$ to the interval $[-h,0)$ by zero, i.e.,
$
x(t)=0$, for $t\in[-h,0].$
This assumption is standard in input-output analysis of time-delay systems and allows the effect of the external perturbation input $a(t)$ on the closed-loop performance to be isolated from the influence of nonzero initial trajectories.

As discussed in the previous section, the residual output $y_{\rm r}(t)$ is used as a monitoring signal, while the performance output $y_{\rm p}(t)$ characterizes the degradation of the closed-loop behavior. 
Consequently, we are interested in quantifying the maximal amplification from the residual output to the performance output for the delayed system \eqref{sys_nonl_cl}. 
A security metric that combines both performance impact and attack detectability was introduced in \cite{Teixeira2015,Teixeira2021}, termed the output-to-output gain (OOG). 
The OOG metric is formulated as the optimal control problem:
\begin{equation}\label{OOG}
OOG\triangleq  \sup\limits_{a\in L_{2e}, x_0=0}||y_{\rm p}||^2_{L_2}, \quad \mbox{s.t.}\quad ||y_{\rm r}||^2_{L_2}\le1,
\end{equation}
where $L_{2e}$ is the extended $L_2$ space, defined as
$$
L_{2e} = \left\{ a: \mR_+ \to \mR^{n_a} \,\big|\, \|a\|_{L_{2[0,T]}} < \infty,\ \forall T < \infty \right\},
$$
which represents all signals that are square integrable over finite periods of time.
Therefore, the OOG characterizes the worst-case performance degradation achievable under the constraint that the residual energy remains bounded.

To estimate the output-to-output gain of the delayed system \eqref{sys_nonl_cl}, we employ a dissipativity-based approach. 
More specifically, we extend the notion of output strict dissipativity from \cite{Khalil}, originally formulated for finite-dimensional systems with Lyapunov storage functions, to time-delay systems described by Lyapunov--Krasovskii functionals. 
This extension allows the infinite-dimensional dynamics induced by the delays to be incorporated directly into the performance analysis framework and provides a convenient basis for deriving delay-dependent LMI conditions for the OOG.
\begin{definition}
The closed-loop time-delay system \eqref{sys_nonl_cl} is said to be {\it output strictly dissipative} with respect to a supply rate $w(x_t,a(t))$ if there exists a positive semidefinite Lyapunov--Krasovskii functional
$
V(t,x_t,\dot x_t),
$
such that the function
\[
\bar V(t):=V(t,x_t,\dot x_t)
\]
is differentiable along the solutions of \eqref{sys_nonl_cl} and satisfies
\begin{equation}\label{def_dissipativity_delay}
\dot{\bar V}(t)
\le
w(x_t,a(t))
-
\gamma y_{\rm p}^{\T}(t)y_{\rm p}(t),
\qquad \gamma>0.
\end{equation}
\end{definition}
\begin{Lemma}\label{Lem_dissipativity_delay}
Assume that the closed-loop time-delay system \eqref{sys_nonl_cl} is output strictly dissipative with respect to the supply rate
$
w(x_t,a(t))
=
y_{\rm r}^{\T}(t)y_{\rm r}(t).
$
Then it is finite-gain $L_2$ output-to-output stable, i.e.,
\begin{equation}\label{OOStab_delay}
\|y_{\rm p}\|_{L_2}^2
\le
\frac{1}{\gamma}
\|y_{\rm r}\|_{L_2}^2
+
\beta(x_0),
\end{equation}
where
$
\beta(x_0)
=
\frac{1}{\gamma}\bar V(0)
\ge0.
$
In addition,
$
OOG\le\frac{1}{\gamma}.
$
\end{Lemma}
\begin{pf}
From \eqref{def_dissipativity_delay}, with
$
w(x_t,a(t))=y_{\rm r}^{\T}(t)y_{\rm r}(t),
$
we obtain
\begin{equation}\label{pf_dissipativity_delay}
\dot{\bar V}(t)
\le
y_{\rm r}^{\T}(t)y_{\rm r}(t)
-
\gamma y_{\rm p}^{\T}(t)y_{\rm p}(t).
\end{equation}
Integrating \eqref{pf_dissipativity_delay} over the interval $[0,T]$ gives
\[
\bar V(T)-\bar V(0)
\le
\int_0^T y_{\rm r}^{\T}(t)y_{\rm r}(t)\,dt
-
\gamma
\int_0^T y_{\rm p}^{\T}(t)y_{\rm p}(t)\,dt .
\]
Therefore,
$$
\int_0^T y_{\rm p}^{\T}(t)y_{\rm p}(t)\,dt
\le
\frac{1}{\gamma}
\int_0^T y_{\rm r}^{\T}(t)y_{\rm r}(t)\,dt 
-\frac{1}{\gamma}\left(\bar V(T)-\bar V(0)\right).
$$
Since the Lyapunov--Krasovskii functional is positive semidefinite, $\bar V(T)\ge0$ for all $T\ge0$. Hence,
\[
\int_0^T y_{\rm p}^{\T}(t)y_{\rm p}(t)\,dt
\le
\frac{1}{\gamma}
\int_0^T y_{\rm r}^{\T}(t)y_{\rm r}(t)\,dt
+
\frac{1}{\gamma}\bar V(0).
\]
Letting $T\to\infty$ yields
\[
\|y_{\rm p}\|_{L_2}^2
\le
\frac{1}{\gamma}
\|y_{\rm r}\|_{L_2}^2
+
\frac{1}{\gamma}\bar V(0),
\]
which gives \eqref{OOStab_delay} with
$
\beta(x_0)=\frac{1}{\gamma}\bar V(0).
$
For zero initial conditions, $\bar V(0)=0$. Therefore, if
$
\|y_{\rm r}\|_{L_2}^2\le1,
$
then
$
\|y_{\rm p}\|_{L_2}^2\le\frac{1}{\gamma}.
$
Taking the supremum over all admissible $a\in L_{2e}$ gives
$
OOG\le\frac{1}{\gamma}.
$
\qed
\end{pf}

To derive tractable conditions guaranteeing the dissipativity result of Lemma~\ref{Lem_dissipativity_delay}, we construct a Lyapunov--Krasovskii functional satisfying \eqref{def_dissipativity_delay}.  
When calculating its derivative, in contrast to a direct substitution of the right-hand side of the state equation instead of $\dot x(t)$, we follow the descriptor approach and treat $\dot x(t)$ as an additional free variable, see \cite{c3,IJRNC15}.
This provides extra degrees of freedom and usually leads to less conservative LMI conditions.

We consider the following Lyapunov--Krasovskii functional:
\begin{equation}\label{LKF_double_delay}
\begin{aligned}
\bar V(t)&=V(t,x_t,\dot x_t)
=x^{\T}(t)Px(t)
+\int_{t-h_{\rm u}}^{t}\!x^{\T}(s)S_{\rm u}x(s)ds\\
&+h_{\rm u}\!\int_{-h_{\rm u}}^{0}\!\int_{t+\theta}^{t}\!\!\dot x^{\T}(s)R_{\rm u}\dot x(s)ds d\theta
+\!\!\int_{t-\tau_{\rm u}(t)}^{t}\!\!x^{\T}(s)Q_{\rm u}x(s)ds\\
&+\int_{t-h_{\rm y}}^{t}\!x^{\T}(s)S_{\rm y}x(s)ds
+h_{\rm y}\!\int_{-h_{\rm y}}^{0}\!\int_{t+\theta}^{t}\dot x^{\T}(s)R_{\rm y}\dot x(s)dsd\theta\\
&+\int_{t-\tau_{\rm y}(t)}^{t}\!x^{\T}(s)Q_{\rm y}x(s)ds,
\end{aligned}
\end{equation}
where
$
P>0,$\,
$S_{\rm u}>0$,\,
$S_{\rm y}>0$,\,
$R_{\rm u}>0$,\,
$R_{\rm y}>0$,\,
$Q_{\rm u}>0$,\,
$Q_{\rm y}>0.$

Taking the derivative of \eqref{LKF_double_delay} along the trajectories of \eqref{sys_nonl_cl} gives
\begin{equation}\label{Vdot_first_bound}
\begin{aligned}
\dot{\bar V}(t)
&\le
2x^{\T}(t)P\dot x(t)
+x^{\T}(t)(S_{\rm u}+S_{\rm y}+Q_{\rm u}+Q_{\rm y})x(t)\\
&
-x^{\T}(t-h_{\rm u})S_{\rm u}x(t-h_{\rm u})
-x^{\T}(t-h_{\rm y})S_{\rm y}x(t-h_{\rm y})\\
&
-(1-d_{\rm u})x^{\T}(t-\tau_{\rm u}(t))Q_{\rm u}x(t-\tau_{\rm u}(t))\\
&
-(1-d_{\rm y})x^{\T}(t-\tau_{\rm y}(t))Q_{\rm y}x(t-\tau_{\rm y}(t))\\
&
+h_{\rm u}^2\dot x^{\T}(t)R_{\rm u}\dot x(t)
+h_{\rm y}^2\dot x^{\T}(t)R_{\rm y}\dot x(t)\\
&
-h_{\rm u}\int_{t-h_{\rm u}}^{t}\dot x^{\T}(s)R_{\rm u}\dot x(s)ds
-h_{\rm y}\int_{t-h_{\rm y}}^{t}\dot x^{\T}(s)R_{\rm y}\dot x(s)ds .
\end{aligned}
\end{equation}
Next, we estimate the integral terms in \eqref{Vdot_first_bound}. 
A standard approach is to apply Jensen's inequality, see, e.g., \cite{chen}. 
However, to obtain less conservative bounds, we instead employ the reciprocally convex approach, see \cite{Park2011}.
For the actuation delay, assume that there exists a matrix $S_{12}^{\rm u}$ such that
$\begin{bmatrix}
R_{\rm u} & S_{12}^{\rm u}\\
* & R_{\rm u}
\end{bmatrix}\ge0.
$
Then,
\begin{multline}\label{reciprocal_u}
-h_{\rm u}\int_{t-h_{\rm u}}^{t}\dot x^{\T}(s)R_{\rm u}\dot x(s)ds
\le
-\begin{bmatrix}
x(t)-x(t-\tau_{\rm u}(t))\\
x(t-\tau_{\rm u}(t))-x(t-h_{\rm u})
\end{bmatrix}^{\!\T}\\
\times
\begin{bmatrix}
R_{\rm u} & S_{12}^{\rm u}\\
* & R_{\rm u}
\end{bmatrix}
\begin{bmatrix}
x(t)-x(t-\tau_{\rm u}(t))\\
x(t-\tau_{\rm u}(t))-x(t-h_{\rm u})
\end{bmatrix}.
\end{multline}
Similarly, for the measurement delay, if there exists a matrix $S_{12}^{\rm y}$ such that
$\begin{bmatrix}
R_{\rm y} & S_{12}^{\rm y}\\
* & R_{\rm y}
\end{bmatrix}\ge0$,
then the corresponding integral term is bounded by
\begin{multline}\label{reciprocal_y}
-h_{\rm y}\int_{t-h_{\rm y}}^{t}\dot x^{\T}(s)R_{\rm y}\dot x(s)ds
\le
-\begin{bmatrix}
x(t)-x(t-\tau_{\rm y}(t))\\
x(t-\tau_{\rm y}(t))-x(t-h_{\rm y})
\end{bmatrix}^{\T}\\
\times
\begin{bmatrix}
R_{\rm y} & S_{12}^{\rm y}\\
* & R_{\rm y}
\end{bmatrix}
\begin{bmatrix}
x(t)-x(t-\tau_{\rm y}(t))\\
x(t-\tau_{\rm y}(t))-x(t-h_{\rm y})
\end{bmatrix}.
\end{multline}

To avoid substituting the right-hand side of \eqref{sys_nonl_cl} directly into $\dot x(t)$, we use the descriptor identity
\begin{equation}\label{descriptor_identity}
\begin{aligned}
0
&=
2\left[x^{\T}(t)P_2^{\T}+\dot x^{\T}(t)P_3^{\T}\right]\\
&\times\!
\left[
Ax(t)+A_1x(t-\tau_{\rm u}(t))
+A_2x(t-\tau_{\rm y}(t))
+Ba(t)-\dot x(t)
\right]\!,
\end{aligned}
\end{equation}
where $P_2$ and $P_3$ are free matrices of appropriate dimensions.

Combining \eqref{Vdot_first_bound}--\eqref{descriptor_identity} with the supply-rate expression, we obtain
\begin{equation}\label{dissipation_estimate}
\dot{\bar V}(t)
-y_{\rm r}^{\T}(t)y_{\rm r}(t)
+\gamma y_{\rm p}^{\T}(t)y_{\rm p}(t)
\le
\eta^{\T}(t)\Psi\eta(t),
\end{equation}
where
$$
\begin{aligned}
\eta^{\T}(t)=
\big[&
x^{\T}(t),\,
\dot x^{\T}(t),\,
x^{\T}(t-h_{\rm u}),\,
x^{\T}(t-h_{\rm y}),\,
x^{\T}(t-\tau_{\rm u}(t)),\,\\
&x^{\T}(t-\tau_{\rm y}(t)),\,
a^{\T}(t)
\big].
\end{aligned}
$$
The matrix $\Psi=\{\Psi_{ij}\}_{i,j=1}^{7}$ is symmetric, with the nonzero upper-triangular blocks given by
\begin{equation}\label{Psi_double_blocks}
\begin{aligned}
\Psi_{11}
&=
A^{\T}P_2+P_2^{\T}A
+S_{\rm u}+S_{\rm y}+Q_{\rm u}+Q_{\rm y}
-R_{\rm u}-R_{\rm y}\\
&\quad
-C_{\rm r}^{\T}C_{\rm r}
+\gamma C_{\rm p}^{\T}C_{\rm p},\\
\Psi_{12}
&=
P-P_2^{\T}+A^{\T}P_3,\,\,
\Psi_{13}
=
S_{12}^{\rm u},\,\,
\Psi_{14}
=
S_{12}^{\rm y},\\
\Psi_{15}
&=
R_{\rm u}-S_{12}^{\rm u}
+P_2^{\T}A_1
+\gamma C_{\rm p}^{\T}C_1,\\
\Psi_{16}
&=
R_{\rm y}-S_{12}^{\rm y}
+P_2^{\T}A_2
-C_{\rm r}^{\T}C_2,\\
\Psi_{17}
&=
P_2^{\T}B
-C_{\rm r}^{\T}D_{\rm r}
+\gamma C_{\rm p}^{\T}D_{\rm p},\\
\Psi_{22}
&=
-P_3-P_3^{\T}
+h_{\rm u}^2R_{\rm u}
+h_{\rm y}^2R_{\rm y},\,\,
\Psi_{23}=0,\\
\Psi_{24}&=0,\,\Psi_{25}
=
P_3^{\T}A_1,\,\,
\Psi_{26}
=
P_3^{\T}A_2,\,\,
\Psi_{27}
=
P_3^{\T}B,\\
\Psi_{33}
&=
-S_{\rm u}-R_{\rm u},\,\,
\Psi_{34}
=
0,\\
\Psi_{35}
&=
R_{\rm u}-(S_{12}^{\rm u})^{\T},\,\,
\Psi_{36}
=
0,\,
\Psi_{37}=0,\\
\Psi_{44}
&=
-S_{\rm y}-R_{\rm y},\,\,
\Psi_{45}
=
0,\,
\Psi_{46}
=
R_{\rm y}-(S_{12}^{\rm y})^{\T},\,\,
\Psi_{47}
=
0,\\
\Psi_{55}
&=
-2R_{\rm u}+S_{12}^{\rm u}+(S_{12}^{\rm u})^{\T}
-(1-d_{\rm u})Q_{\rm u}
+\gamma C_1^{\T}C_1,\\
\Psi_{56}
&=
0,\,\,
\Psi_{57}
=
\gamma C_1^{\T}D_{\rm p},\\
\Psi_{66}
&=
-2R_{\rm y}+S_{12}^{\rm y}+(S_{12}^{\rm y})^{\T}
-(1-d_{\rm y})Q_{\rm y}
-C_2^{\T}C_2,\\
\Psi_{67}
&=
-C_2^{\T}D_{\rm r},\,\,
\Psi_{77}
=
-D_{\rm r}^{\T}D_{\rm r}
+\gamma D_{\rm p}^{\T}D_{\rm p}.
\end{aligned}
\end{equation}
All omitted blocks are determined by symmetry.
Therefore, if $\Psi\le0$, then from \eqref{dissipation_estimate} we obtain 
the dissipativity inequality \eqref{def_dissipativity_delay} required in Lemma~\ref{Lem_dissipativity_delay}. 

This leads to the following result.
\begin{Theorem}\label{Thm_double_delay_OOG}
Consider the time-delay closed-loop system \eqref{sys_nonl_cl}. 
Assume that there exist matrices
$
P>0$,\,
$S_{\rm u}>0$,\,
$S_{\rm y}>0$,\,
$R_{\rm u}>0$,\,
$R_{\rm y}>0$,\,
$Q_{\rm u}>0$,\,
$Q_{\rm y}>0$,\,
$P_2$,\, $P_3$,\, $S_{12}^{\rm u}$,\, and $S_{12}^{\rm y}$ of appropriate dimensions, and a scalar $\gamma>0$ such that the following LMIs are feasible:
\begin{equation}\label{LMI_double_delay_main}
\Psi\le0,\quad
\begin{bmatrix}
R_{\rm u} & S_{12}^{\rm u}\\
* & R_{\rm u}
\end{bmatrix}\ge0,\quad
\begin{bmatrix}
R_{\rm y} & S_{12}^{\rm y}\\
* & R_{\rm y}
\end{bmatrix}\ge0,
\end{equation}
where $\Psi$ is defined by \eqref{Psi_double_blocks}.
Then the system is finite-gain $L_2$ output-to-output stable and
$
OOG\le \frac{1}{\gamma}.
$
\end{Theorem}
\begin{Remark}
The best upper bound provided by Theorem~\ref{Thm_double_delay_OOG} is obtained by maximizing $\gamma$ subject to the LMIs \eqref{LMI_double_delay_main}.
\end{Remark}

\section{Special cases}

In this section, we consider two special cases of the general delay-dependent framework developed above, where the delays are constant. 
For this setting, finite-dimensional approximations based on Pad\'e expansions can be employed, allowing the delayed system to be represented approximately by a finite-dimensional delay-free model. 
We then discuss a class of systems admitting an exact finite-dimensional reduction following the approach of \cite{Churilov13}.

\subsection{Pad\'e approximation}
We now consider the special case where the delays are constant:
$
\tau_{\rm u}(t)=h_{\rm u},
\tau_{\rm y}(t)=h_{\rm y}.
$
In this case, the closed-loop system \eqref{sys_nonl_cl} becomes
\begin{equation}\label{sys_const_delay_2}
\begin{aligned}
\dot x(t)
&=
Ax(t)
+A_1x(t-h_{\rm u})
+A_2x(t-h_{\rm y})
+Ba(t),\\
y_{\rm p}(t)
&=
C_{\rm p}x(t)
+C_1x(t-h_{\rm u})
+D_{\rm p}a(t),\\
y_{\rm r}(t)
&=
C_{\rm r}x(t)
+C_2x(t-h_{\rm y})
+D_{\rm r}a(t).
\end{aligned}
\end{equation}

Although the system \eqref{sys_const_delay_2} remains infinite-dimensional due to the delayed state terms, constant delays admit finite-dimensional approximations in the frequency domain. 
A common approach is to approximate each delay operator, $\e^{-sh_{\rm u}}$ and $\e^{-sh_{\rm y}}$, by rational transfer functions using Pad\'e approximation. 
For simplicity, consider the first-order Pad\'e approximation
\begin{equation}\label{pade_first_general}
\e^{-sh}
\approx
\frac{1-\frac{h}{2}s}{1+\frac{h}{2}s}.
\end{equation}
Higher-order approximations can also be employed to improve approximation accuracy, at the expense of increasing the system dimension.

To approximate the delayed states, introduce auxiliary variables
\[
z_{\rm u}(t)\approx x(t-h_{\rm u}),
\qquad
z_{\rm y}(t)\approx x(t-h_{\rm y}),
\]
satisfying
\begin{equation}\label{pade_aux_general}
\left(1+\frac{h_{\rm u}}{2}\frac{d}{dt}\right)z_{\rm u}(t)
=
\left(1-\frac{h_{\rm u}}{2}\frac{d}{dt}\right)x(t),
\end{equation}
and
\begin{equation}\label{pade_aux_general2}
\left(1+\frac{h_{\rm y}}{2}\frac{d}{dt}\right)z_{\rm y}(t)
=
\left(1-\frac{h_{\rm y}}{2}\frac{d}{dt}\right)x(t).
\end{equation}
Equivalently,
\begin{equation}\label{pade_state_general}
\dot z_{\rm u}(t)
=
-\frac{2}{h_{\rm u}}z_{\rm u}(t)
+\frac{2}{h_{\rm u}}x(t)
-\dot x(t),
\end{equation}
and
\begin{equation}\label{pade_state_general2}
\dot z_{\rm y}(t)
=
-\frac{2}{h_{\rm y}}z_{\rm y}(t)
+\frac{2}{h_{\rm y}}x(t)
-\dot x(t).
\end{equation}
Define the augmented state
$
\chi(t)=
\begin{bmatrix}
x(t)\\
z_{\rm u}(t)\\
z_{\rm y}(t)
\end{bmatrix}.
$
Then substituting \eqref{pade_state_general}--\eqref{pade_state_general2} into \eqref{sys_const_delay_2} yields the approximate finite-dimensional system
\begin{equation}\label{sys_pade_compact_general}
\begin{aligned}
\dot\chi(t)
&=
\mathcal A\chi(t)+\mathcal B a(t),\\
y_{\rm p}(t)
&=
\mathcal C_{\rm p}\chi(t)+D_{\rm p}a(t),\\
y_{\rm r}(t)
&=
\mathcal C_{\rm r}\chi(t)+D_{\rm r}a(t),
\end{aligned}
\end{equation}
where
\[
\mathcal A=
\begin{bmatrix}
A & A_1 & A_2\\
\frac{2}{h_{\rm u}}I-A & -\frac{2}{h_{\rm u}}I-A_1 & -A_2\\
\frac{2}{h_{\rm y}}I-A & -A_1 & -\frac{2}{h_{\rm y}}I-A_2
\end{bmatrix},
\]
\[
\mathcal B=
\begin{bmatrix}
B\\
-B\\
-B
\end{bmatrix},\quad
\mathcal C_{\rm p}=
\begin{bmatrix}
C_{\rm p} & C_1 & 0
\end{bmatrix},
\quad
\mathcal C_{\rm r}=
\begin{bmatrix}
C_{\rm r} & 0 & C_2
\end{bmatrix}.
\]

Therefore, after applying the Pad\'e approximation, the original infinite-dimensional delayed system is transformed into the finite-dimensional system \eqref{sys_pade_compact_general}. 
As a result, the output-to-output gain can be analyzed using the OOG framework for delay-free linear systems, yielding tractable LMI conditions for estimating the worst-case output-to-output amplification. 
\begin{Theorem}[\cite{Teixeira2021}]\label{thm_Andre}
Assume there exist $P\ge0$ and $\gamma>0$ such that the LMI
\begin{equation}\label{fd_oog_lmi}
\begin{bmatrix}
\mathcal A^{\T}P+P\mathcal A & P\mathcal B\\
* & 0
\end{bmatrix}
-
\begin{bmatrix}
\mathcal C_{\rm r}^{\T}\\
D_{\rm r}^{\T}
\end{bmatrix}
\begin{bmatrix}
\mathcal C_{\rm r} & D_{\rm r}
\end{bmatrix}
+
\gamma
\begin{bmatrix}
\mathcal C_{\rm p}^{\T}\\
D_{\rm p}^{\T}
\end{bmatrix}
\begin{bmatrix}
\mathcal C_{\rm p} & D_{\rm p}
\end{bmatrix}
\le0
\end{equation}
holds. Then the output-to-output gain of the approximate system \eqref{sys_pade_compact_general} satisfies the bound
$OOG\le \frac{1}{\gamma}.$
\end{Theorem}

%
%

\begin{Remark}
Although Pad\'e approximations provide a convenient finite-dimensional representation of constant-delay systems, the resulting model remains only approximate. Moreover, the approximation accuracy may deteriorate for large delays, low approximation orders, or in the presence of fast transient dynamics. In addition, Pad\'e-based models are primarily applicable to systems with constant delays and become significantly more involved in the case of time-varying delays. A detailed investigation of how the Pad\'e approximation affects the output-to-output gain of the original system may be a topic for future research.

For these reasons, the main results of this paper were derived directly for the original delay system \eqref{sys_nonl_cl}. The proposed Lyapunov--Krasovskii framework applies to systems with distinct measurement and actuation delays, allows the delays to be time varying, and provides delay-dependent conditions without introducing approximation errors associated with finite-dimensional reductions.

Nevertheless, the Pad\'e approximation possesses an important advantage that complements the proposed delay-dependent analysis. While the LMI conditions developed in this paper provide upper bounds on the output-to-output gain of the original time-delay system, they do not directly yield perturbation signals that approach the worst-case performance. In contrast, after replacing the delay operator by a finite-dimensional Pad\'e realization, the resulting model falls within the class of delay-free systems considered in \cite{Teixeira2021}. Consequently, the generalized eigenvalue formulation from \cite{Teixeira2013} can be employed to construct suboptimal perturbation signals. These signals can then be applied to the original delayed system to assess the practical tightness of the obtained OOG estimates.

Therefore, the Pad\'e approach serves a dual purpose. On the one hand, it provides a finite-dimensional approximation that can be analyzed using standard OOG tools. On the other hand, it offers a systematic mechanism for constructing perturbation signals tailored to the presence of delays. This feature is particularly useful in numerical studies, where direct optimization over the infinite-dimensional delayed dynamics is generally difficult.
\end{Remark}

\subsection{Finite-dimension reducible systems}
We now consider a special case in which the constant-delay system admits an exact finite-dimensional reduction. 
This class of systems is known as finite-dimension reducible (FD-reducible). 
In contrast to Pad\'e approximation, the reduction is exact: after an initial transient interval of length equal to the delay, the state trajectory of the delay system satisfies a finite-dimensional delay-free differential equation.

To unify the subsequent analysis, we consider the closed-loop system \eqref{sys_nonl_cl} with constant delays and rewrite it in the generic form
\begin{equation}\label{sys_fd_delay}
\begin{aligned}
\dot x(t)
&=
Ax(t)+A_{\rm d}x(t-h)+Ba(t),\\
y_{\rm p}(t)
&=
C_{\rm p}x(t)+C_{\rm d,p}x(t-h)+D_{\rm p}a(t),\\
y_{\rm r}(t)
&=
C_{\rm r}x(t)+C_{\rm d,r}x(t-h)+D_{\rm r}a(t),
\end{aligned}
\end{equation}
where the matrices $A_{\rm d}$, $C_{\rm d,p}$ and $C_{\rm d,r}$ depend on the considered delay configuration.
In particular, three cases will be considered.
\begin{itemize}
\item {\it Equal measurement and actuation delays}, i.e.,
$
\tau_{\rm u}=\tau_{\rm y}=h.
$
In this case,
$
A_{\rm d}=A_1+A_2,
C_{\rm d,p}=C_1,
C_{\rm d,r}=C_2.
$
\item {\it Measurement delay only}, i.e.,
$
\tau_{\rm u}=0,
\tau_{\rm y}=h.
$
Since the actuation delay becomes delay-free, its contribution is absorbed into the nominal dynamics. Therefore,
$
A\leftarrow A+A_1,
$
and
$
A_{\rm d}=A_2,
C_{\rm d,p}=0,
C_{\rm d,r}=C_2.
$
\item {\it Actuation delay only}, i.e.,
$
\tau_{\rm u}=h,
\tau_{\rm y}=0.
$
Similarly, the measurement delay contribution is absorbed into the nominal dynamics:
$
A\leftarrow A+A_2,
$
and
$
A_{\rm d}=A_1,
C_{\rm d,p}=C_1,
C_{\rm d,r}=0.
$
\end{itemize}

The following result extends the FD-reducibility conditions of \cite{Churilov13} and \cite{Yamalova15} to the forced system \eqref{sys_fd_delay}.
\begin{Theorem}[FD-reducibility with external input]\label{Thm_FD_forced}$\,$\\
Consider the forced delay system
\begin{equation}\label{fd_forced_system}
\dot x(t)=Ax(t)+A_{\rm d}x(t-h)+Ba(t),
\end{equation}
where $x(t)\in\mathbb R^{n_x}$, $a(t)\in\mathbb R^{n_a}$, and $h>0$ is constant. 
Assume that
\begin{equation}\label{fd_cond_state_input}
A_{\rm d}A^kA_{\rm d}=0,
\quad
A_{\rm d}A^kB=0,
\quad
k=0,1,\ldots,n_x-1.
\end{equation}
Equivalently,
\begin{equation}\label{fd_cond_exp}
A_{\rm d}e^{At}A_{\rm d}=0,
\quad
A_{\rm d}e^{At}B=0,
\quad
\forall t\in\mathbb R.
\end{equation}
Then every solution of \eqref{fd_forced_system} satisfies, for $t\ge h$,
\begin{equation}\label{fd_reduced_system}
\dot x(t)=Dx(t)+Ba(t),
\end{equation}
where
$D=A+A_{\rm d}e^{-Ah}$.
\end{Theorem}
\begin{pf}
First, we show the equivalence between \eqref{fd_cond_state_input} and \eqref{fd_cond_exp}.
Using the expansion
$
e^{At}=\sum_{k=0}^{\infty}\frac{A^k t^k}{k!},
$
we obtain
$
A_{\rm d}e^{At}A_{\rm d}
=
\sum_{k=0}^{\infty}
\frac{A_{\rm d}A^kA_{\rm d}}{k!}t^k.
$
Thus, by the Cayley--Hamilton theorem,
\[
A_{\rm d}A^kA_{\rm d}=0,\qquad k=0,\ldots,n_x-1,
\]
implies
\[
A_{\rm d}e^{At}A_{\rm d}=0,\qquad \forall t\in\mathbb R.
\]
Conversely, differentiating the above expression at $t=0$ gives
\[
A_{\rm d}A^kA_{\rm d}=0,
\qquad k=0,1,\ldots .
\]
The same argument applies to the input matrix $B$. 

It remains to show that, under \eqref{fd_cond_exp}, the solution satisfies the delay-free equation \eqref{fd_reduced_system} for $t\ge h$.
For $t\ge h$, the variation-of-constants formula applied on the interval $[t-h,t]$ gives
\[
\begin{aligned}
x(t-h)
&=
e^{-Ah}x(t)
-
\int_{t-h}^{t}
e^{A(t-h-s)}A_{\rm d}x(s-h)\,ds\\
&\quad
-
\int_{t-h}^{t}
e^{A(t-h-s)}Ba(s)\,ds .
\end{aligned}
\]
Multiplying both sides by $A_{\rm d}$ and using \eqref{fd_cond_exp}, we obtain
\[
A_{\rm d}x(t-h)
=
A_{\rm d}e^{-Ah}x(t).
\]
Substituting this identity into \eqref{fd_forced_system} yields \eqref{fd_reduced_system}. \qed
\end{pf}

\begin{Remark}[Block representation]\label{Rem_FD_block}
The conditions in \eqref{fd_cond_state_input} can be related to the block representation used in \cite{Churilov13}. 
In particular, the homogeneous FD-reducibility condition is equivalent to the existence of an invertible matrix $S$ such that
\begin{equation}\label{fd_block1}
S^{-1}AS=
\begin{bmatrix}
U&0\\
W&V
\end{bmatrix},
\qquad
S^{-1}A_{\rm d}S=
\begin{bmatrix}
0&0\\
\bar W&0
\end{bmatrix}.
\end{equation}
For the forced system, a simple structural condition ensuring
$
A_{\rm d}e^{At}B=0$ for all $t$,
is
\begin{equation}\label{fd_block2}
S^{-1}B=
\begin{bmatrix}
0\\
B_2
\end{bmatrix}.
\end{equation}
for some $B_2$.
Indeed, in these coordinates,
$
S^{-1}e^{At}B
=
\begin{bmatrix}
0\\
e^{Vt}B_2
\end{bmatrix},
$
and therefore
$
S^{-1}A_{\rm d}e^{At}B
=
\begin{bmatrix}
0&0\\
\bar W&0
\end{bmatrix}
\begin{bmatrix}
0\\
e^{Vt}B_2
\end{bmatrix}
=0.
$
More generally, if
$
S^{-1}B=
\begin{bmatrix}
B_1\\
B_2
\end{bmatrix},
$
then the exact equivalent block condition is
$
\bar W e^{Ut}B_1=0,
$
or equivalently
$
\bar W U^kB_1=0, k=0,1,\ldots,n_U-1.
$

Thus, the conditions \eqref{fd_block1}-\eqref{fd_block2} is a convenient sufficient structural condition, while the exponential conditions in \eqref{fd_cond_exp} gives the exact characterization.
\end{Remark}

\begin{corollary}[Delay-free output representation]\label{Cor_FD_delay_free_outputs}
Suppose \\ that the assumptions of Theorem~\ref{Thm_FD_forced} hold and 
\begin{equation}\label{fd_condition_output}
C_{\rm d,p}D^kB=0,
\quad
C_{\rm d,r}D^kB=0,
\quad
k=0,1,\ldots,n_x-1.
\end{equation}
Equivalently, $
C_{\rm d,p}e^{Dt}B=0,
\,
C_{\rm d,r}e^{Dt}B=0$,
for all $t\in\mathbb R$.
Then
the outputs admit the delay-free representation
\begin{equation}\label{fd_delay_free_outputs}
\begin{aligned}
y_{\rm p}(t)
&=
\bar C_{\rm p}x(t)+D_{\rm p}a(t),\\
y_{\rm r}(t)
&=
\bar C_{\rm r}x(t)+D_{\rm r}a(t),
\end{aligned}
\qquad t\ge 2h,
\end{equation}
where
\begin{equation}\label{fd_bar_C}
\bar C_{\rm p}
=
C_{\rm p}+C_{\rm d,p}e^{-Dh},
\qquad
\bar C_{\rm r}
=
C_{\rm r}+C_{\rm d,r}e^{-Dh}.
\end{equation}
\end{corollary}
\begin{pf}
For $t\ge 2h$, both $t$ and $t-h$ belong to the interval where the state satisfies the reduced finite-dimensional equation
\[
\dot x(t)=Dx(t)+Ba(t).
\]
Using the variation-of-constants formula on $[t-h,t]$, we obtain
$
x(t-h)
=
e^{-Dh}x(t)
-
\int_{t-h}^{t} e^{D(t-h-s)}Ba(s)\,ds.
$
Multiplying this identity by $C_{\rm d,p}$ gives
\[
C_{\rm d,p}x(t-h)
=
C_{\rm d,p}e^{-Dh}x(t)
-
\int_{t-h}^{t}C_{\rm d,p}e^{D(t-h-s)}Ba(s)\,ds.
\]
By \eqref{fd_condition_output}, the integral term vanishes. Hence,
$
C_{\rm d,p}x(t-h)
=
C_{\rm d,p}e^{-Dh}x(t).
$
The same argument with $C_{\rm d,r}$ gives
$
C_{\rm d,r}x(t-h)
=
C_{\rm d,r}e^{-Dh}x(t).
$
Substituting these identities into the output equations yields \eqref{fd_delay_free_outputs}. \qed
\end{pf}

Under the additional output conditions of Corollary~\ref{Cor_FD_delay_free_outputs}, the original delay system admits the finite-dimensional delay-free representation
\begin{equation}\label{fd_final_system}
\begin{aligned}
\dot x(t)&=Dx(t)+Ba(t),\\
y_{\rm p}(t)&=\bar C_{\rm p}x(t)+D_{\rm p}a(t),\\
y_{\rm r}(t)&=\bar C_{\rm r}x(t)+D_{\rm r}a(t),
\end{aligned}
\qquad t\ge 2h.
\end{equation}

Therefore, the systems \eqref{sys_fd_delay} and \eqref{fd_final_system} are equivalent only for $t\ge 2h$. 
However, for zero initial conditions, the outputs are invariant on $h$, which is shown in the following corollary.

\begin{corollary}[Delay invariance of the OOG]
\label{Cor_FD_delay_invariance}
Assume that the nominal delayed system is FD-reducible, i.e., the assumptions of Theorem~\ref{Thm_FD_forced} and Corollary~\ref{Cor_FD_delay_free_outputs} hold.
Then, for zero initial conditions, the input-output maps
\[
a\mapsto y_{\rm p},
\qquad
a\mapsto y_{\rm r}
\]
are independent of the delay value $h$. Consequently, the output-to-output gain is independent of $h$.
\end{corollary}

\begin{pf}Following Remark~\ref{Rem_FD_block}, introduce the FD coordinates
\[
v(t)=S^{-1}x(t)
=
\begin{bmatrix}
v_1(t)\\
v_2(t)
\end{bmatrix}.
\]
In these coordinates, the forced delayed system becomes
\[
\begin{aligned}
\dot v_1(t)&=Uv_1(t),\\
\dot v_2(t)&=Wv_1(t)+Vv_2(t)+\bar Wv_1(t-h)+B_2a(t).
\end{aligned}
\]
Since $v(0)=S^{-1}x(0)=0$, we have
$
v_1(t)\equiv0, t\ge0,
$
and the lower subsystem reduces to
\[
\dot v_2(t)=Vv_2(t)+B_2a(t).
\]
Hence the state component generated by the perturbation is independent of the delay value $h$.

Now, consider the performance output
\[
y_{\rm p}(t)
=
C_{\rm p}x(t)+C_{{\rm d},p}x(t-h)+D_{\rm p}a(t).
\]
Since for zero initial conditions
\[
x(t-h)
=
\int_{0}^{t-h}e^{D(t-h-s)}Ba(s)\,ds,
\]
from $C_{\rm d,p}e^{Dt}B=0$, we  obtain
\[
C_{{\rm d},p}x(t-h)
=
\int_{0}^{t-h}
C_{{\rm d},p}e^{D(t-h-s)}Ba(s)\,ds
=0.
\]
Consequently, the performance output reduces to
\[
y_{\rm p}(t)=C_{\rm p}x(t)+D_{\rm p}a(t),
\]
and this expression is independent of $h$, since $x(t)$ is generated by the $h$-independent dynamics.

The same argument applies to the residual output.
Thus, both output maps $a\mapsto y_{\rm p}$ and $a\mapsto y_{\rm r}$ are independent of the delay. Since the OOG is defined only in terms of the $L_2$ energies of these outputs, $OOG$ is independent of $h$. \qed
\end{pf}

Corollary~\ref{Cor_FD_delay_invariance} significantly simplifies the performance analysis problem. In particular, once the FD-reducibility conditions are satisfied, the delay value itself no longer affects the OOG. Therefore, there is no need to introduce a separate FD-reduced performance metric, and the original OOG can be computed directly from the finite-dimensional representation using the delay-free framework of \cite{Teixeira2021}. 

The FD coordinates also provide additional structural insight into how perturbation channels affect the system. In particular, they allow one to identify perturbation directions that are harmless from a performance perspective and, conversely, configurations that may be particularly vulnerable to stealthy attacks. This observation is formalized in the following remarks.

\begin{Remark}[Protected configurations]\label{Rem_Spec}
The FD structure naturally allows one to identify perturbation channels that do not degrade the performance output.
Assume that the conditions of Corollary~\ref{Cor_FD_delay_invariance} hold, so that, for the purpose of OOG analysis, the delayed system is equivalent to the delay-free subsystem
\[
\dot v_2(t)=Vv_2(t)+B_2a(t).
\]
Let
$
y_{\rm p}(t)=\widehat C_{\rm p}v_2(t)+D_{\rm p}a(t)$,
$
y_{\rm r}(t)=\widehat C_{\rm r}v_2(t)+D_{\rm r}a(t),
$
where
$
\widehat C_{\rm p}=C_{\rm p}S
\begin{bmatrix}
0\\
I
\end{bmatrix}$,
$
\widehat C_{\rm r}=C_{\rm r}S
\begin{bmatrix}
0\\
I
\end{bmatrix}.
$

A protected configuration is obtained when the admissible perturbation channels do not affect the performance output. This occurs if
\[
\widehat C_{\rm p}V^kB_2=0,
\qquad
D_{\rm p}=0,
\qquad
k=0,1,\ldots,n_V-1.
\]
Then, for zero initial conditions,
$
y_{\rm p}(t)\equiv0
$
for every perturbation $a(t)$, i.e., these perturbations do not degrade the performance output. If, at the same time,
\[
\widehat C_{\rm r}V^kB_2\neq0
\quad
\text{for some }k,
\]
or
$
D_{\rm r}\neq0,
$
then the same perturbations remain visible in the residual output and can be detected.

\end{Remark}

\begin{Remark}[Vulnerable configurations]
\label{Rem_FD_observer}
Theoretically, one may also consider the opposite, vulnerable configuration, in which the admissible perturbation channels are invisible from the residual output while still affecting the performance output. This occurs if
\[
\widehat C_{\rm r}V^kB_2=0,
\qquad
D_{\rm r}=0,
\qquad
k=0,1,\ldots,n_V-1,
\]
while
$
\widehat C_{\rm p}V^kB_2\neq0$
%
or
$
D_{\rm p}\neq0.
$
In this case,
$
y_{\rm r}(t)\equiv0,
$
whereas
$
y_{\rm p}(t)\not\equiv0,
$
implying an unbounded OOG.

However, such vulnerable configurations are difficult to realize within observer-based architectures. Indeed, perturbations entering the measurement channel are directly visible in the residual output through the feedthrough term
$
D_{\rm r}=\Gamma_{\rm y}.
$
On the other hand, actuator or physical perturbations remove this direct residual feedthrough, but generally either violate the FD-reducibility condition
$
A_{\rm d}A^kB=0,
$
or become visible in the residual output through the closed-loop dynamics. Therefore, observer-based FD-reducible systems are naturally associated with protected configurations rather than vulnerable ones.
\end{Remark}

Summarizing the above discussion, the FD-reducibility framework should be viewed as a special structural case rather than a generic analysis tool. Although the required algebraic conditions are restrictive and are typically difficult to satisfy for general observer-based control architectures, they provide valuable structural insight when they hold. In particular, FD-reducibility not only yields an exact finite-dimensional representation and a computationally efficient OOG analysis, but also reveals which communication channels should be protected in order to preserve both performance and detectability.

\section{Numerical example}

In this section, we demonstrate how communication delays can affect the output-to-output gain and increase the impact of external perturbations while maintaining a small residual output.

Consider the system \eqref{sys_lin}--\eqref{obs_lin} with the parameters
\begin{equation}\label{sys_ex}
\begin{aligned}
A_{\rm p} & =
\begin{bmatrix}
1&-2&-1\\
0&-0.5&0\\
0&0&-0.1
\end{bmatrix},
\qquad
B_{\rm p}=
\begin{bmatrix}
0\\
1\\
1
\end{bmatrix},
\\
C_{\rm m,o} &=
\begin{bmatrix}
1&0&0\\
0&0&1
\end{bmatrix},
\qquad
C_{\rm p,o}=
\begin{bmatrix}
0&1&0
\end{bmatrix},
\qquad
D_{\rm p,o}=1.
\end{aligned}
\end{equation}
The performance output is $y_{\rm p}=x_2+\tilde u$, while the measured output is
$
y_{\rm m}
=
\begin{bmatrix}
x_1\\
x_3
\end{bmatrix}.
$

The controller and observer gains are chosen as
$
L=[2.43,\,-3.24,\,-0.66],
$
and
$
K=
\begin{bmatrix}
3&0\\
-1&0\\
0&0.9
\end{bmatrix},
$
such that the eigenvalues of the matrices
$
A_{\rm p}+B_{\rm p}L
$
and
$
A_{\rm p}-KC_{\rm m,o}
$
are equal to
$
\{-1,-2,-0.5\}.
$

Assume that the perturbation affects only the second measurement channel:
$
\Gamma_{\rm y}
=
\begin{bmatrix}
0\\
1
\end{bmatrix},
$
while
$
\Gamma_{\rm u}=0$,
$
\Gamma_{\rm p}=\begin{bmatrix}
0\\
0\\
0
\end{bmatrix}.
$

For the delay-free system, the output-to-output gain can be computed using the  framework of \cite{Teixeira2021}, yielding $OOG=2.007$.

For the time-delay case, the upper bound of the output-to-output gain can be computed using the LMI conditions of Theorem~1. Figure~\ref{h_OOG_3D} illustrates the obtained OOG bound as a function of the measurement ($h_{\rm y}$) and actuation ($h_{\rm u}$) delays in the constant-delay case. As expected, the OOG increases with both delays, indicating that communication delays may significantly amplify the impact of external perturbations while maintaining a small residual output. It is worth noting that for $h_{\rm u}=h_{\rm y}=0$, the proposed framework yields OOG=2.007, which exactly coincides with the result obtained using the delay-free framework. 

\begin{figure}
\centering
\includegraphics[width=3.5in]{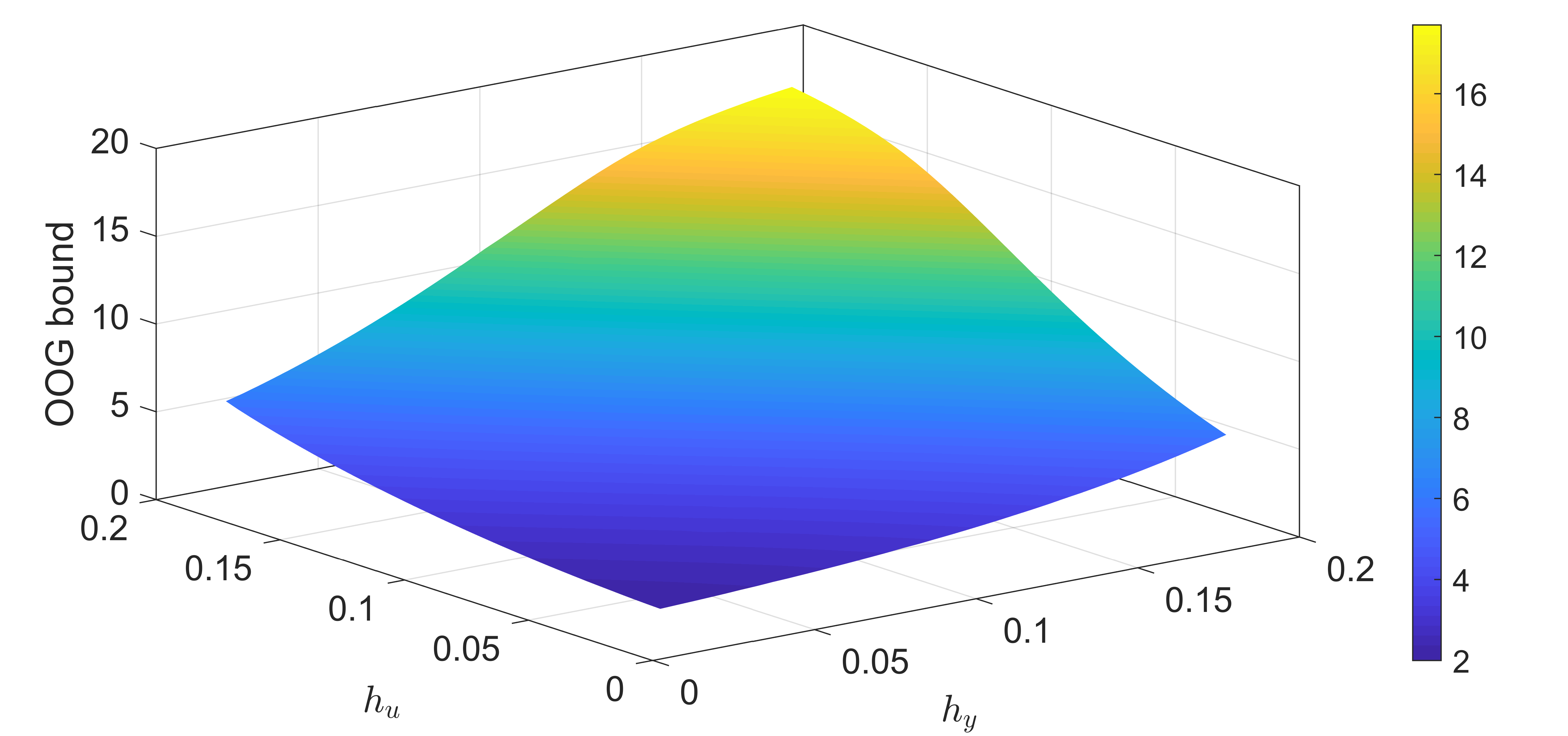}
\caption{Upper bound on the output-to-output gain obtained from Theorem~1 as a function of the constant delays $h_{\rm u}$ and $h_{\rm y}$.}
\label{h_OOG_3D}
\end{figure}

Next, we investigate the influence of slowly varying delays. Figure~\ref{h_OOG_same_stable} shows the OOG bound for the case $h_{\rm u}=h_{\rm y}=h$ and several values of the delay variation parameter $d$. One can observe that moderate delay variations have only a limited influence on the obtained estimate. In particular, the OOG bound is primarily determined by the nominal delay value $h$, while increasing the admissible delay derivative results only in a relatively small increase of the estimate. Therefore, for slowly varying delays, the constant-delay case provides a good indication of the overall system behavior.

\begin{figure}
\centering
\includegraphics[width=3.5in]{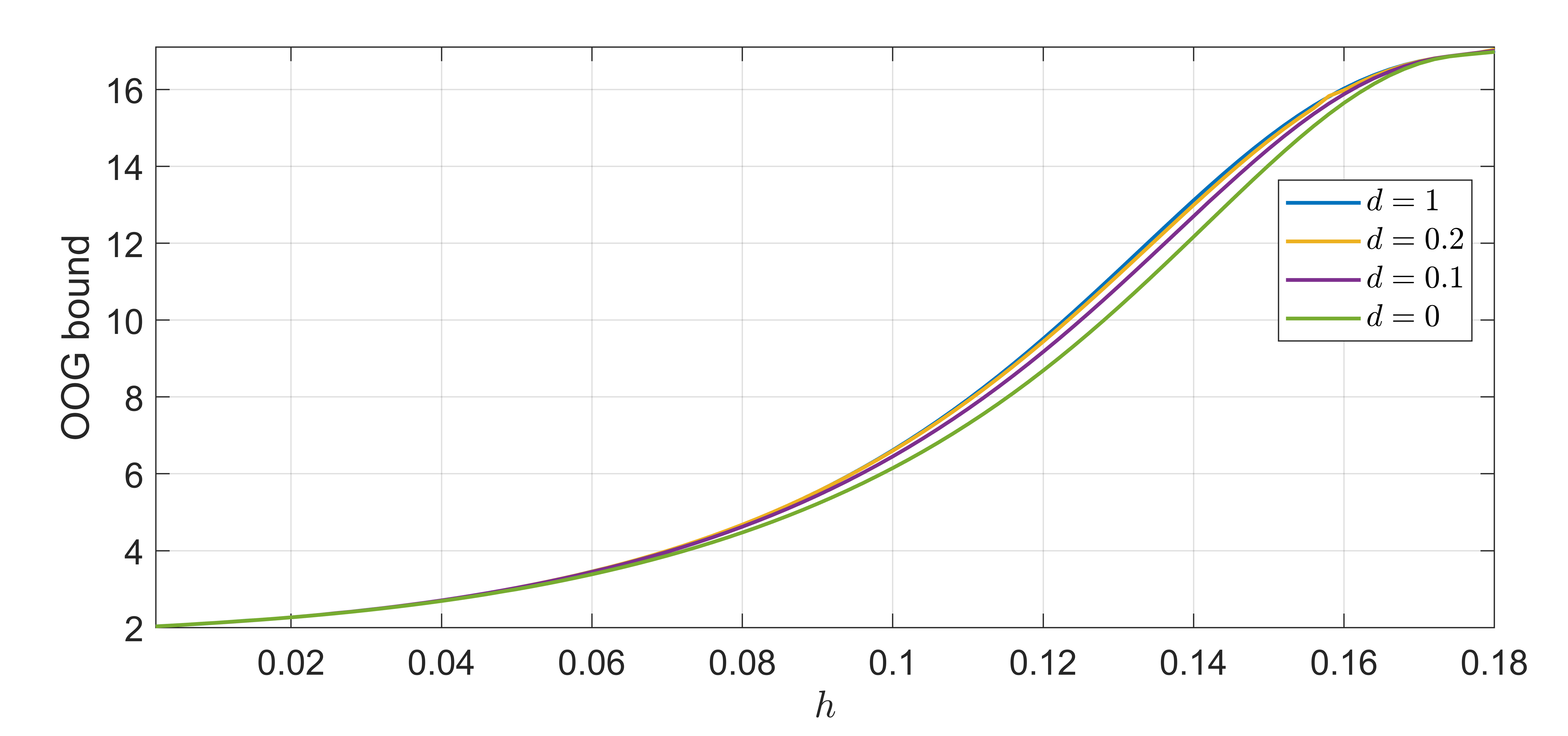}
\caption{OOG bound for equal delays ($h_{\rm u}=h_{\rm y}=h$) and different values of the delay variation parameter $d$.}
\label{h_OOG_same_stable}
\end{figure}

To illustrate the practical impact of delays, we next construct a perturbation signal that approximately maximizes the output-to-output gain.
For the delay-free case, a suboptimal perturbation signal $a(t)$ can be obtained by constructing a sampled-data approximation of the closed-loop system, see \cite{SeifullaevIFAC26}. Taking sampling period $t_s=0.1$ and finite horizon $T=15$, we define
\[
{\bf a}
=
[a[0]^{\T},\ldots,a[N]^{\T}]^{\T},
\]
and similarly
\[
{\bf y}_{\rm p}
=
[y_{\rm p}[0]^{\T},\ldots,y_{\rm p}[N]^{\T}]^{\T},
\,\,\,
{\bf y}_{\rm r}
=
[y_{\rm r}[0]^{\T},\ldots,y_{\rm r}[N]^{\T}]^{\T}.
\]
Then there exist matrices $\mathcal T_{\rm p}$ and $\mathcal T_{\rm r}$ such that
\[
{\bf y}_{\rm p}
=
\mathcal T_{\rm p}{\bf a},
\qquad
{\bf y}_{\rm r}
=
\mathcal T_{\rm r}{\bf a}.
\]
The OOG problem can therefore be approximated as
\begin{equation}\label{OOG_appr_ex}
\sup_{{\bf a}}
{\bf a}^{\T}
\mathcal T_{\rm p}^{\T}
\mathcal T_{\rm p}
{\bf a},
\qquad
\text{s.t.}
\qquad
{\bf a}^{\T}
\mathcal T_{\rm r}^{\T}
\mathcal T_{\rm r}
{\bf a}
\le1.
\end{equation}
This problem can be rewritten as a generalized eigenvalue problem and solved efficiently. The resulting perturbation signal
$
a(t)=a_{\rm y}(t)=\mathrm{ZOH}({\bf a})
$
is shown in Figure~\ref{subopt_a_delay}.
\begin{figure}
\centering
\includegraphics[width=3.5in]{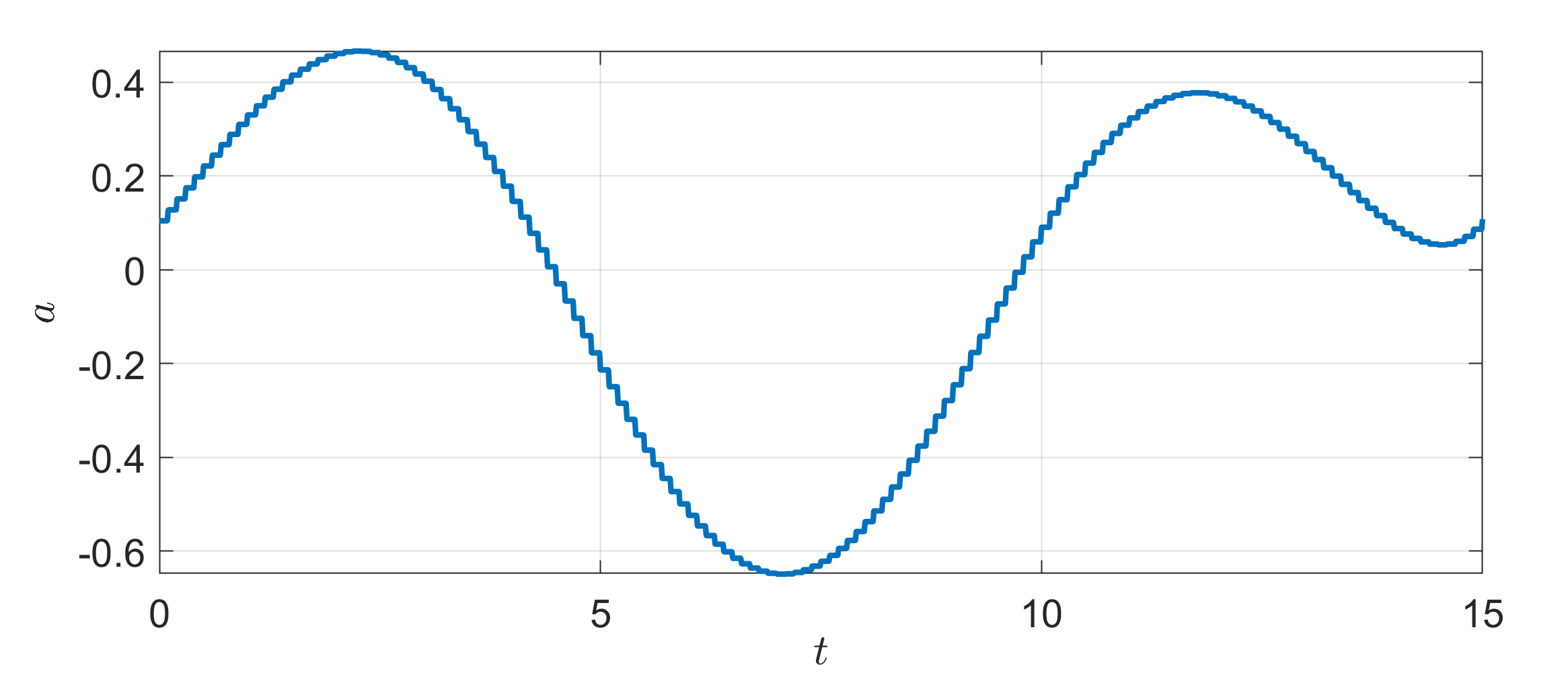}
\caption{Suboptimal perturbation signal obtained from \eqref{OOG_appr_ex}.}
\label{subopt_a_delay}
\end{figure}

In the time-delay case, however, constructing an optimal perturbation signal becomes considerably more challenging due to the infinite-dimensional nature of the dynamics. To overcome this difficulty, we employ the Pad\'e approximation. The delay operators are replaced by finite-dimensional Pad\'e realizations, resulting in an augmented delay-free system. The generalized eigenvalue approach can then be applied to this approximating model to obtain a perturbation signal adapted to the presence of delays. Since Figure~\ref{h_OOG_same_stable} indicates that slowly varying delays produce estimates close to those of the constant-delay case, the resulting signal may also be viewed as a reasonable suboptimal perturbation for systems with small delay variations.

Figures~\ref{energy_h01} and \ref{energy_h005} show the accumulated energies of the performance and residual outputs for the original delayed system when the perturbation signal obtained from the Pad'e approximation is applied. 
For $h_{\rm u}=h_{\rm y}=0.1$, the obtained numerical value of the output-to-output gain is approximately 3.08, while the LMI-based estimate is 6.1. For $h_{\rm u}=h_{\rm y}=0.05$, the numerical value decreases to approximately 2.4, whereas the corresponding estimate is 2.95. Thus, for small delays, the proposed Lyapunov--Krasovskii conditions provide a remarkably accurate estimate of the actual performance degradation. As the delays increase, the gap between the numerical value and the theoretical upper bound naturally becomes larger, reflecting the conservatism introduced by the Lyapunov-based analysis. Nevertheless, the obtained estimates remain informative and correctly capture the substantial increase of the output-to-output gain caused by communication delays.
\begin{figure}
\centering
\includegraphics[width=3.5in]{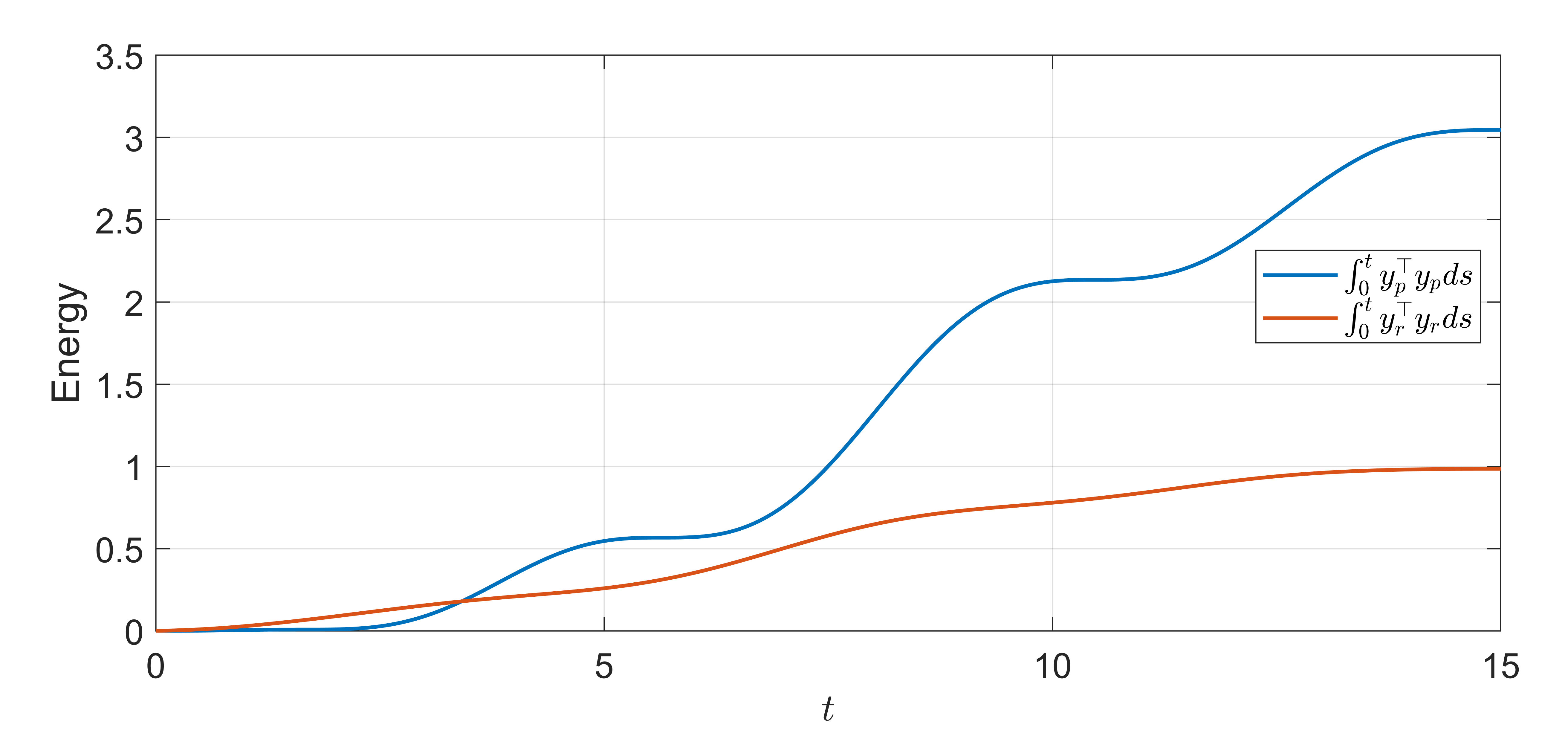}
\caption{Accumulated performance and residual output energies for $h_{\rm u}=h_{\rm y}=0.1$. The resulting numerical OOG is approximately 3.08.}
\label{energy_h01}
\end{figure}
\begin{figure}
\centering
\includegraphics[width=3.5in]{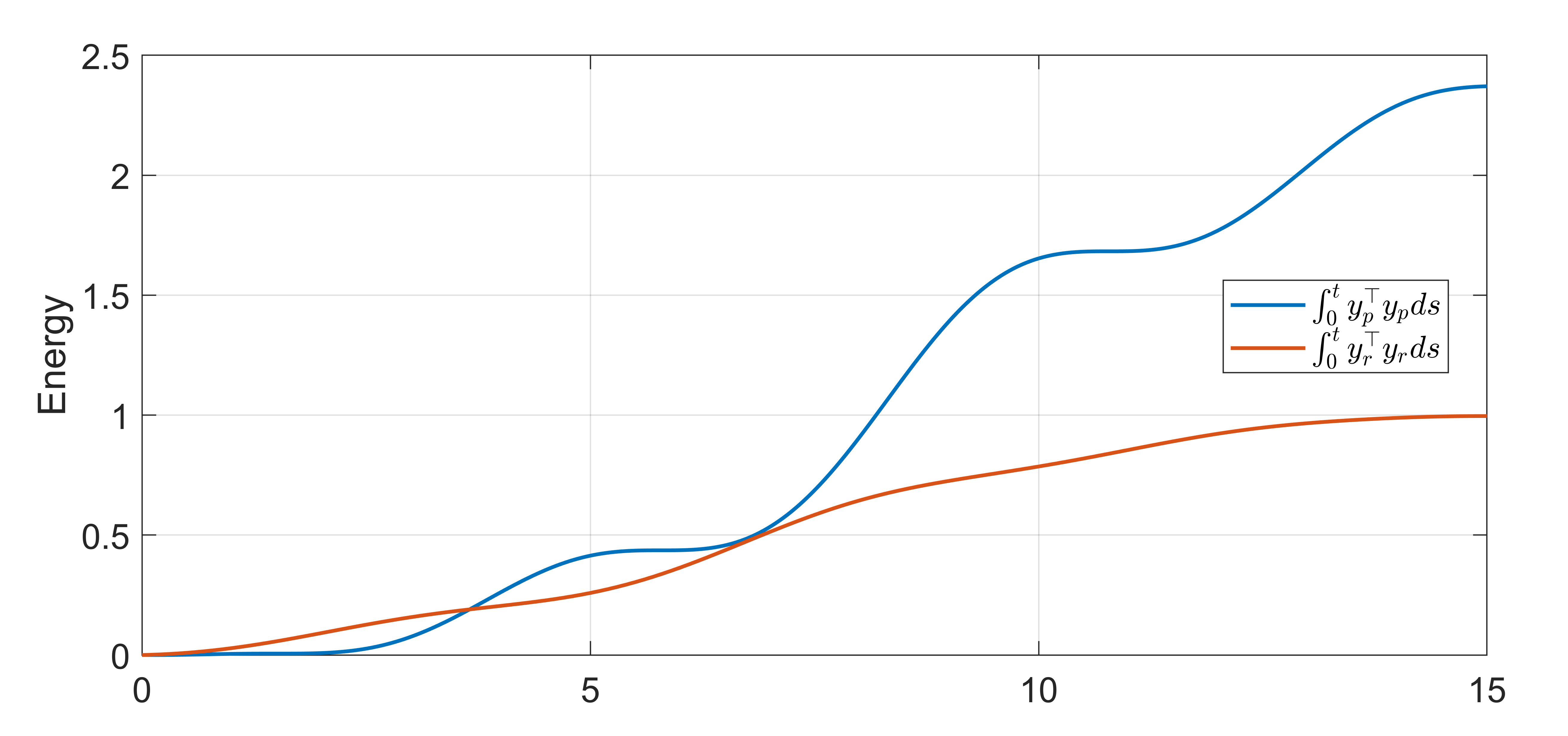}
\caption{Accumulated performance and residual output energies for $h_{\rm u}=h_{\rm y}=0.05$. The resulting numerical OOG is approximately 2.4.}
\label{energy_h005}
\end{figure}

\subsection{Finite-dimension reducible system}
Finally, we illustrate the FD-reducible case. Unlike the preceding example, the goal here is not to demonstrate the degradation of the OOG caused by delays. Indeed, under the FD-reducibility conditions developed above, the OOG becomes independent of the delay value. Instead, the FD framework provides additional structural insight into how perturbation channels propagate through the delayed closed-loop system.
In particular, the FD coordinates naturally reveal which communication channels should be protected and which admissible perturbations remain detectable without degrading performance. Next, we construct a simple observer-based example illustrating this protected configuration.

Let
$
A_{\rm p}=
\begin{bmatrix}
-1&0\\
1&-2
\end{bmatrix}$,
$B_{\rm p}=
\begin{bmatrix}
1\\
0
\end{bmatrix}$,
$C_{\rm m,o}=
\begin{bmatrix}
1&1
\end{bmatrix}$,
$C_{\rm p,o}=
\begin{bmatrix}
1&0
\end{bmatrix}$,
$D_{\rm p,o}=0,
$
and choose
$
L=
\begin{bmatrix}
0.5&0
\end{bmatrix}$,
$K=
\begin{bmatrix}
0\\
1
\end{bmatrix}.
$
Consider the measurement-delay-only case, i.e.,
$
\tau_{\rm u}=0$, $\tau_{\rm y}=h$.
Assume that the perturbation can affect only the measurement channel:
$
\Gamma_{\rm y}=1,
\Gamma_{\rm u}=0,
\Gamma_{\rm p}=0.
$

Since the actuation delay is absent, its contribution is absorbed into the delay-free matrix. Hence,
\[
A=
\begin{bmatrix}
A_{\rm p}+B_{\rm p}L & -B_{\rm p}L\\
KC_{\rm m,o} & A_{\rm p}-KC_{\rm m,o}
\end{bmatrix}
=
\begin{bmatrix}
-0.5&0&-0.5&0\\
1&-2&0&0\\
0&0&-1&0\\
1&1&0&-3
\end{bmatrix},
\]
and
$
A_{\rm d}
=
\begin{bmatrix}
0&0\\
-KC_{\rm m,o}&0
\end{bmatrix}
=
\begin{bmatrix}
0&0&0&0\\
0&0&0&0\\
0&0&0&0\\
-1&-1&0&0
\end{bmatrix}.
$
The perturbation input matrix is
$
B = \begin{bmatrix}
B_{\rm p}\Gamma_{\rm u}&0&\Gamma_{\rm p}\\
B_{\rm p}\Gamma_{\rm u}&-K\Gamma_{\rm y}&\Gamma_{\rm p}
\end{bmatrix},
=
\begin{bmatrix}
0&0&0\\
0&0&0\\
0&0&0\\
0&-1&0
\end{bmatrix}.
$
Then, we obtain that the system is already in the FD block form with \(S=I\):
$
A=
\begin{bmatrix}
U&0\\
W&V
\end{bmatrix}$,
$
A_{\rm d}=
\begin{bmatrix}
0&0\\
\bar W&0
\end{bmatrix}$,
$
B=
\begin{bmatrix}
0\\
B_2
\end{bmatrix}.
$

The corresponding output matrices are
$
C_{\rm p}=
\begin{bmatrix}
1&0&0&0
\end{bmatrix}$,
$
C_{\rm d,p}=0$,
$
D_{\rm p}=0,
$
and
$
C_{\rm r}=
\begin{bmatrix}
-C_{\rm m,o}&C_{\rm m,o}
\end{bmatrix}
=
\begin{bmatrix}
-1&-1&1&1
\end{bmatrix},
$
$
C_{\rm d,r}=
\begin{bmatrix}
C_{\rm m,o}&0
\end{bmatrix}
=
\begin{bmatrix}
1&1&0&0
\end{bmatrix}$,
$
D_{\rm r}=1.
$
Since \(C_{\rm d,p}=0\), the output reducibility condition for the performance output holds trivially. Moreover, because \(B\) acts only on the lower FD block and \(C_{\rm d,r}\) acts only on the upper FD block, we have
$
C_{\rm d,r}D^kB=0,
k=0,1,2,3.
$
Thus, the output conditions of Corollary~\ref{Cor_FD_delay_free_outputs} are satisfied.

Finally, following Remark~\ref{Rem_Spec}, the reduced performance and residual output matrices are
$
\widehat C_{\rm p}
=
C_{\rm p}
\begin{bmatrix}
0\\
1
\end{bmatrix}
=0$,
$
\widehat C_{\rm r}
=
C_{\rm r}
\begin{bmatrix}
0\\
1
\end{bmatrix}
=1.
$
Since \(D_{\rm p}=0\), the admissible perturbation does not affect the performance output, i.e.,
$
y_{\rm p}(t)\equiv0
$
for zero initial conditions. At the same time, since \(\widehat C_{\rm r}\neq0\) and \(D_{\rm r}=1\), the perturbation remains visible in the residual output. Hence this example represents a protected configuration: the remaining admissible perturbation channel does not degrade performance, while its effect appears in the residual signal.

\section{Conclusions}
This paper addressed the performance analysis of observer-based networked control systems with communication delays using the output-to-output gain framework. By extending dissipativity-based OOG analysis to time-delay systems, delay-dependent LMI conditions were derived using Lyapunov--Krasovskii functionals, descriptor techniques, and reciprocally convex inequalities. These conditions provide explicit upper bounds on the OOG for systems with independent delays in measurement and actuation channels. In addition, two special classes of constant-delay systems were investigated: Pad\'e-approximated systems and finite-dimension reducible systems, both of which admit finite-dimensional OOG analysis under suitable assumptions.
\bibliography{references}

@ARTICLE{Sandberg15,
  author={Sandberg, H. and Amin, S. and Johansson, K.H.},
  journal={IEEE Control Systems Magazine}, 
  title={Cyberphysical Security in Networked Control Systems: An Introduction to the Issue}, 
  year={2015},
  volume={35},
  number={1},
  pages={20-23}
  }

@ARTICLE{c3,
    author ={E. Fridman},
    title ={A refined input delay approach to sampled-data control},
    journal = {Automatica}, volume = 46, year = 2010, pages ={421--427}}

@BOOK{chen,
    author = "Gu, K. and Kharitonov, V. and  Chen, J.",
    title = {Stability of time-delay systems},
    address = {Boston},
    publisher = "Birkh{\"a}user",
    year = 2003}

@ARTICLE{IJRNC15,
    author ={Seifullaev, R.E. and Fradkov, A.L.},
    title ={Robust nonlinear sampled-data system analysis based on {F}ridman's method and {S}-procedure},
    journal = {International Journal of Robust and Nonlinear Control}, volume = 26, year = 2016, pages ={201--217}}

@BOOK{Khalil,
    author = "Khalil, Hassan K.",
    title = {Nonlinear Systems},
    publisher = "Prentice Hall PTR",
    year = 2002}

@book{Fridman_book,
	author={E.~Fridman},
	title={Introduction to Time-Delay Systems: Analysis and Control},
	publisher={Birkh{\"a}user},
	year={2014}
	}

@INPROCEEDINGS{Cardenas08,
  author={Cardenas, A. A. and Amin, S. and Sastry, Sh.},
  booktitle={the 28th International Conference on Distributed Computing Systems Workshops}, 
  title={Secure Control: Towards Survivable Cyber-Physical Systems}, 
  year={2008},
  volume={},
  number={},
  pages={495-500},
  }

@INPROCEEDINGS{Churilov13,
  author={Churilov, Alexander and Medvedev, Alexander and Mattsson, Per},
  booktitle={52nd IEEE Conference on Decision and Control}, 
  title={Finite-dimensional reducibility of time-delay systems under pulse-modulated feedback}, 
  year={2013},
  volume={},
  number={},
  pages={2078-2083},
 }

@article{Park2011,
title = {Reciprocally convex approach to stability of systems with time-varying delays},
journal = {Automatica},
volume = {47},
number = {1},
pages = {235-238},
year = {2011},
author = {PooGyeon Park and Jeong Wan Ko and Changki Jeong},
}

@BOOK{Zhou96,
  title={Robust and Optimal Control},
  author={Zhou, K and Doyle, J C and Glover, K},
  year={1996},
  publisher={Prentice-Hall, Inc.}
}

@article{WANG07,
title = {An {LMI} approach to ${H}_-$ index and mixed ${H}_-/{H}_\infty$ fault detection observer design},
journal = {Automatica},
volume = {43},
number = {9},
pages = {1656-1665},
year = {2007},
author = {Jian Liang Wang and Guang-Hong Yang and Jian Liu},
}

@Inbook{Teixeira2013,
author="Teixeira, A. M. H.
and Sou, K. Ch.
and Sandberg, H.k
and Johansson, K. H.",
editor="Tarraf, D. C.",
title="Quantifying Cyber-Security for Networked Control Systems",
bookTitle="Control of Cyber-Physical Systems",
year="2013",
publisher="Springer International Publishing",
address="Heidelberg",
pages="123--142"
}

@INPROCEEDINGS{Teixeira2015,
  author={Teixeira, A. M. H. and Sandberg, H. and Johansson, K. H.},
  booktitle={54th IEEE Conference on Decision and Control (CDC)}, 
  title={Strategic stealthy attacks: The output-to-output $l_2$-gain}, 
  year={2015},
  pages={2582--2587},
}

@incollection{Teixeira2021,
author="A. M. H. Teixeira",
editor="R. M. Ferrari and A. M. H. Teixeira",
title="Security metrics for control systems",
booktitle="Safety, Security and Privacy for Cyber-Physical Systems",
year="2021",
publisher="Springer International Publishing",
}

@article{Yamalova15,
  title={Finite-Dimensional Hybrid Observer for Delayed Impulsive Model of Testosterone Regulation},
  author={Yamalova, Diana and Churilov, Alexander and Medvedev, Alexander},
  journal={Mathematical Problems in Engineering},
  volume={2015},
  number={1},
  pages={190463},
  year={2015},
  publisher={Wiley Online Library}
}

@INPROCEEDINGS{SeifullaevIFAC26,
  title={Impact analysis of hidden faults in nonlinear control systems using output-to-output gain},
  author={Seifullaev, Ruslan and Teixeira, Andr{\'e}},
  booktitle = {23rd IFAC World Congress (accepted)}, 
  year={2026}
}

@BOOK{Kharitonov13,
    author={Kharitonov,V. L.},
   title ={Time-Delay Systems. Lyapunov Functions and Matrices}, publisher={Birkhauser},
   year= 2013}

@article{Richard03,
title = {Time-delay systems: an overview of some recent advances and open problems},
journal = {Automatica},
volume = {39},
number = {10},
pages = {1667-1694},
year = {2003},
author = {Jean-Pierre Richard},
}

@ARTICLE{24_Mousavinejad,
  author={Mousavinejad, E. and Yang, F. and Han, Q.-L. and Vlacic, L.},
  journal={IEEE Transactions on Cybernetics}, 
  title={A Novel Cyber Attack Detection Method in Networked Control Systems}, 
  year={2018},
  volume={48},
  number={11},
  pages={3254-3264},
  }

@article{ZHAO2020109128,
title = {A dynamic event-triggered approach to observer-based PID security control subject to deception attacks},
journal = {Automatica},
volume = {120},
pages = {109128},
year = {2020},
author = {Di Zhao and Zidong Wang and Guoliang Wei and Qing-Long Han}}

\end{document}